\documentclass[12pt]{article}
\usepackage{graphicx,psfrag,epsf}
\usepackage{enumerate}
\usepackage{natbib}

\newcommand{\blind}{0}

\usepackage{amsthm,amsmath}
\usepackage[colorlinks,citecolor=blue,urlcolor=blue,filecolor=blue,backref=page]{hyperref}
\usepackage[utf8]{inputenc}

\usepackage{amssymb}
\usepackage{verbatim} 
\usepackage{color}
\usepackage{caption}
\usepackage{times}
\usepackage{bm}
\usepackage[plain,noend]{algorithm2e}
\usepackage{setspace}
\usepackage{booktabs}

\usepackage{epstopdf}% To incorporate .eps illustrations using PDFLaTeX, etc.
\usepackage[caption=false]{subfig}% Support for small, `sub' figures and tables
\theoremstyle{plain}% Theorem-like structures provided by amsthm.sty
\newtheorem{theorem}{Theorem}[section]
\newtheorem{lemma}[theorem]{Lemma}

\theoremstyle{definition}

\theoremstyle{remark}

\DeclareMathOperator{\logit}{logit}

\begin{document}

\def\spacingset#1{\renewcommand{\baselinestretch}%
{#1}\small\normalsize} \spacingset{1}

%%%%%%%%%%%%%%%%%%%%%%%%%%%%%%%%%%%%%%%%%%%%%%%%%%%%%%%%%%%%%%%%%%%%%%%%%%%%%%

\if0\blind
{
  \title{\bf Posterior Averaging Information Criterion}
  \author{Shouhao Zhou \thanks{
    Email: szhou@mdanderson.org}\hspace{.2cm}\\
    Department of Public Health Sciences, Pennsylvania State University}
  \maketitle
} \fi

\if1\blind
{
  \bigskip
  \bigskip
  \bigskip
  \begin{center}
    {\LARGE\bf Title}
\end{center}
  \medskip
} \fi

\bigskip
\begin{abstract}
We propose an innovative model selection method, the posterior averaging information criterion, for Bayesian model assessment from a predictive perspective.  The theoretical foundation is built on the Kullback-Leibler divergence to quantify the similarity between the proposed candidate model and the underlying true model.  From a Bayesian perspective, our method evaluates the candidate models over the entire posterior distribution in terms of predicting a future independent observation.  Without assuming that the true distribution is contained in the candidate models, the new criterion is developed by correcting the asymptotic bias of the posterior mean of the log-likelihood against its expected log-likelihood.  It can be generally applied even for Bayesian models with degenerate non-informative prior.  The simulation in both normal and binomial settings demonstrates an outstanding small sample performance.  %A real data example illustrates possible disagreement between explanatory and predictive model selection approaches..
\end{abstract}

\noindent%

\spacingset{1.45}
\section{Introduction}

Model selection plays a key role in applied statistical practice.  A clearly defined model selection criterion or score usually lies at the heart of any statistical model selection procedure, and facilitates the comparison of competing models through the assignment of some sort of preference or ranking to the alternatives.  Standard criteria include adjusted $R^2$ (Wherry, 1931), Akaike information criterion (AIC; Akaike, 1973), minimum description length (MDL; Rissanen, 1978) and Schwarz information criterion (SIC; Schwarz, 1978), to name but a few. 

To choose a proper criterion for a statistical data analysis project, it is essential to distinguish the ultimate goal of modeling.  Geisser and Eddy (1979) challenged research workers with two fundamental questions that should be asked in advance of any procedure conducted for model selection:\vspace{-.05in}
\begin{enumerate}
	\item Which of the models best explains a given set of data? 
	\item Which of the models yields the best predictions for future observations from the same process that generated the given set of data?\vspace{-.05in}
\end{enumerate}
The former question, which concerns the accuracy of the model in describing the current data, has been an empirical problem for many years. It represents the explanatory perspective. The latter question, which represents the predictive perspective, concerns the accuracy of the model in predicting future data, having drawn substantial attention in recent decades. If an infinitely large quantity of data is available, the predictive perspective and the explanatory perspective may not differ significantly.  However, with a limited number of observations, as we encounter in practice, it is challenging for predictive model selection methods to achieve an optimal balance between goodness of fit and parsimony.  

A substantial group of predictive model selection criteria were proposed based on the Kullback-Leibler information divergence (Kullback and Leibler, 1951), an objective measure to estimate the overall closeness of a probability distribution and the underlying true model. 
%The true model $f$ in this scenario is referred to as the \textit{unknown} data generating mechanism.  
On both theoretical and applied fronts, Kullback-Leibler divergence in model selection has drawn a huge amount of attention, and a large related body of literature now exists for both frequentist and Bayesian inference.

%The availability of both fast computers and advanced numerical methods %in recent years enables the empirical popularity of Bayesian modeling, %which allows the additional flexibility to incorporate the information %out of the data, represented by the prior distribution.

Bayesian approaches to statistical inference have specific concerns regarding the interpretation of parameters and models.  
%Accordingly, model selection methods for Bayesian modeling are supposed to be different from those used in frequentist approaches.  
However, most of the Kullback-Leibler based Bayesian criteria follow essentially the frequentist paradigm insofar as they select a model using plug-in estimators of the parameters.  Starting from the Bayesian predictive information criterion (BPIC; Ando, 2007), model selection criteria were developed over the entire posterior distribution.  Nevertheless, BPIC has a number of limitations, particularly with asymmetric posterior distributions.  Furthermore, BPIC is undefined under improper prior distributions, while the expected penalized loss assumes that the true model contained in the approximating family, which limits its use in practice.

%In this article we investigate the Bayesian philosophy in Bayesian model %choice and propose a versatile method, PAIC, to compare different %Bayesian statistical models for the same data.  Our approach is %consistent with the Kullback-Leibler information measure and bases %itself on prediction of a future unobservable.  The benchmark of the %comparison is the underlying unobserved true distribution, and Bayesian %models are evaluated over the entire posterior distribution of %parameters.

%In this work we focus on model selection . 
To explain the intuition of the proposed Bayesian criterion, in Section \ref{sec2} we review the Kullback-Leibler divergence, its application and development in frequentist statistics and the adaption to Bayesian modeling based on plug-in parameter estimation.  In Section \ref{sec3}, major attention is given to the Kullback-Leibler based predictive criterion for models evaluated by averaging over the posterior distributions of parameters. A generally applicable method, the posterior averaging information criterion (PAIC), is proposed for comparing different Bayesian statistical models under regularity conditions. Our criterion is developed by correcting the asymptotic bias of using the posterior mean of the log-likelihood as an estimator of its expected log-likelihood, and we prove that the asymptotic property holds even though the candidate models are misspecified.  In Section \ref{sec4} we present some numerical studies in both normal and binomial cases to investigate its performance with small sample sizes. %We also provide a real data variable selection example in Section \ref{sec5} to exhibit possible differences between the explanatory and predictive approaches.  
We conclude with a few summary remarks and discussions in Section \ref{sec6}.

\section{Kullback-Leibler divergence and model selection}\label{sec2}
Kullback and Leibler (1951) derived an information measure to assess the directed `distance' between any two models. If we assume that $f(\tilde{y})$ and $g(\tilde{y})$ respectively represent the probability density distributions of the `true model' and the `approximate model' on the same measurable space, the Kullback-Leibler divergence is defined by
$$I(f,g)=\int f(\tilde{y})\cdot \log \frac{f(\tilde{y})}{g(\tilde{y})}d\tilde{y}= E_{\tilde{y}} [\log f(\tilde{y})]- E_{\tilde{y}}[\log g(\tilde{y})],$$
which is always non-negative, reaching the minimum value of $0$ when $f$ is the same as $g$ almost surely. It is interpreted as the `information' lost when $g$ is used to approximate $f$. Namely, the smaller the value of $I(f,g)$, the closer we consider the model $g$ to be to the true distribution.

Only the second term of $I(f,g)$ is relevant in practice to compare different possible models without full knowledge of the true distribution. This is because the first term, $E_{\tilde{y}} [\log f(\tilde{y})]$, is a constant that depends on only the unknown true distribution $f$, and can be neglected in model comparison for given data. 

Let $\textbf{y}=(y_1,y_2,\cdots,y_n)$ be $n$ independent observations in the data and $\tilde{y}$, an unknown but potentially observable quantity, represents a future independent observation has the same probability density function $f(\tilde{y})$, and an approximate model $m$ with density $g_m(\tilde{y}|\theta^m)$ among a list of potential models $m=1,2,\cdots,M$.  For notational purposes, we ignore the model index $m$ when there is no ambiguity.  The true model $f$ is referred to as the \textit{unknown} data generating mechanism, not necessarily to be encompassed in the approximate model family. 

As $n\rightarrow\infty$, the average of the log-likelihood $$\frac{1}{n}L(\theta | \textbf{y})=\frac{1}{n}\sum_{i=1}^{n}{\log g(y_i|\theta)}$$ tends to $E_{\tilde{y}} [\log g(\tilde{y} | \theta )]$ by the law of large numbers, which suggests how we can estimate the second term of $I(f,g)$.  

The model selection based on the Kullback-Leibler divergence is straightforward when all the operating models are fixed probability distributions, i.e., $g(\tilde y|\theta)=g(\tilde y)$.  The model with the largest empirical log-likelihood $\sum_i \log g(y_i)$ is favored.  However, when the distribution family $g(\tilde{y}|\theta)$ contains some unknown parameters $\theta$, the model fitting should be done first so that we may know what values the free parameters will probably take, given the data.  Therefore, the log-likelihood is not optimal for the predictive modeling, when the data were used twice in both model fitting and evaluation. For a desirable out-of-sample predictive performance, a common idea is to identify a bias correction term to rectify the over-estimation bias of the in-sample estimate. 

In the frequentist setting, the general model selection procedure chooses candidate models specified by some point estimate $\hat{\theta}$ based on a certain statistical principle such as maximum likelihood.   A considerable amount of theoretical research has addressed this problem by correcting for the bias of $\frac{1}{n} \sum_i \log g(y_i| \hat\theta)$ in estimation of $E_{\tilde{y}} [\log g(\tilde{y} | \hat \theta )]$ (Akaike, 1973; Takeuchi, 1976; Hurvich and Tsai, 1989; Murata et al., 1994; Konishi and Kitagawa, 1996).  A nice review can be found in Burnham and Anderson (2002).

Since the introduction of the Akaike Information Criterion (AIC; Akaike, 1973), researchers have commonly applied frequenstist model selection methods into Bayesian modeling.  However, the differences in the underlying philosophies between Bayesian and frequentist statistical inference caution against such direct applications.  There also have been a few attempts to specialize the Kullback-Leibler divergence for Bayesian model selection (Geisser and Eddy, 1979; San Martini and Spezzaferri, 1984; Laud and Ibrahim, 1995) in the last century.  Such methods are limited either in the scope of methodology or computational feasibility, especially when the parameters of the Bayesian models are in high-dimensional hierarchical structures.

The seminal work of Spiegelhalter et al. (2002, 2014) proposed Deviance Information Criterion (DIC) as a Bayesian adaption to AIC and implemented it within Bayesian inference using Gibbs sampling (BUGS; Spiegelhalter et al., 1994).  Although the estimation lacks a theoretical foundation (Meng and Vaida, 2006; Celeux et al., 2006a), $-\textsc{dic}/2n$, as a model selection criterion, heuristically estimates $E_{\tilde{y}} [\log g(\tilde{y} | \bar \theta )]$, the expected out-of-sample log-likelihood specified at the posterior mean, after assuming that the proposed model encompasses the true model.   Alternative methods can be found either using a similar approach for mixed-effects models (Vaida and Blanchard, 2005; Liang et al., 2009; Donohue et al. 2011) or using numerical approximation (Plummer, 2008) to estimate cross-validative predictive loss (Efron, 1983).  %Reviews of those methods can be found in Kadane and Lazar (2004) and Vehtari and Ojanen (2012).     

\section{Posterior averaging information criterion}
\label{sec3}
\subsection{Posterior averaged discrepancy function for model selection}
The preceding methods in general can be viewed as Bayesian adaptation of the information criteria originally designed for frequentist statistics, when each model is assessed in terms of the similarity between the true distribution $f$ and the model density function specified by the plug-in parameters.  This may not be ideal since, in contrast to frequentist modeling, ``Bayesian inference is the process of fitting a probability model to a set of data and summarizing the result by \textit{a probability distribution on the parameters} of the model and on unobserved quantities such as predictions for new observations'' (Gelman et al., 2003).  Rather than considering a model specified by a point estimate, it is more reasonable to assess the goodness of a Bayesian model in terms of prediction against the posterior distribution.     

Obtaining the posterior averaged Kullback-Leibler discrepancy, rather than the Kullback-Leibler discrepancy specified at some point estimate, could be more computationally intensive, requiring a large set of posterior samples for numerical averaging when analytical form is not available.  However, advanced computer technology developed in recent years has made this computational cost much more feasible for Bayesian model selection.  Reviews on recent developments can be found in the next section.

\subsection{Posterior averaging information criterion}

Ando (2007) proposed an estimator for the posterior averaged discrepancy function, $$\eta=E_{\tilde{y}}[E_{\theta | \textbf{y}}\log g(\tilde{y}| \theta )].$$  Under certain regularity conditions, it was shown that an asymptotic unbiased estimator of $ \eta$ is
\begin{eqnarray}
\hat \eta ^{BPIC} 
&=& \frac1n E_{\theta | \textbf{y}} \log L (\theta | \textbf{y}) - \frac1n [\ E_{\theta | \textbf{y}}\log \{ \pi (\theta) L(\theta | \textbf{y}) \} - log \{\pi (\hat{\theta}) L(\hat{\theta}| \textbf{y})\} \nonumber \\
 &&            + tr\{J_{n}^{-1}(\hat{\theta}) I_n(\hat{\theta})\}+\frac{K}{2}] \label{eq:Ando.orig} \\
&\triangleq& \frac1n E_{\theta | \textbf{y}} \log L (\theta | \textbf{y}) - BC_1 \nonumber \\ 
&=& \frac1n log \{\pi (\hat{\theta}) L(\hat{\theta}| \textbf{y})\}- \frac1n [\ E_{\theta | \textbf{y}}\log \pi (\theta ) + tr\{J_{n}^{-1}(\hat{\theta}) I_n(\hat{\theta})\}+\frac{K}{2}] \label{eq:Ando} \\ 
&\triangleq& \frac1n log \{\pi (\hat{\theta}) L(\hat{\theta}| \textbf{y})\} - BC_2. \nonumber
\end{eqnarray}
Here, $BC$ denotes the bias correction term, $\hat{\theta}$ is the posterior mode, $K$ is the cardinality of $\theta$, and matrices $J_n$ and $I_n$ are some empirical estimators for the Bayesian asymptotic Hessian matrix, 
$$J(\theta )=-E_{\tilde{y}}\left(\frac{\partial ^{2}\log \{g(\tilde{y}| \theta )\pi_{0}(\theta )\}}{\partial \theta \partial \theta ^{\prime }}\right)$$
and Bayesian asymptotic Fisher information matrix,
$$I(\theta )=E_{\tilde{y}}\left(\frac{\partial \log \{g(\tilde{y}| \theta )\pi_{0}(\theta )\}}{\partial \theta }\frac{\partial \log \{g(\tilde{y}| \theta
)\pi _{0}(\theta )\}}{\partial \theta ^{\prime }}\right),$$
where $\log \pi_0(\theta)=\lim_{n\rightarrow\infty}n^{-1}\log \pi(\theta)$.

The BPIC is introduced as $-2n\cdot \hat \eta ^{BPIC}$ and applicable when the true model $f$ is not necessarily in the specified family of probability distributions. In model comparison, the candidate model with a minimum BPIC value is favored. However, it has the following limitations in practice. 

\noindent 1.  Equation (\ref{eq:Ando.orig}) was from the original presentation for BPIC in Ando (2007). After simple math cancelling out the term $ \frac1n E_{\theta | \textbf{y}} \log L (\theta|\textbf{y})$ in both estimator and bias correction term, it was actually the plug-in estimate $\frac1n log \{\pi (\hat{\theta}) L(\hat{\theta}|\textbf{y})\}$, as shown in equation (\ref{eq:Ando}), in estimation of $\eta$ with some bias correction.  Compared with the natural estimator $n^{-1} E_{\theta|\textbf{y}} \log [L(\theta|\textbf{y})]$, the estimation efficiency of $\eta$ using plug-in estimator is suboptimal when the posterior distribution is asymmetric or with non-zero correlation between parameters, which occurs in a majority of cases in Bayesian modeling. This will be further illustrated in our simulation studies when we compare the bias correction performance of various criteria in small sample size. 

\noindent 2. The BPIC cannot be calculated when the prior distribution $\pi (\theta )$ is degenerate, a situation that commonly occurs in Bayesian analysis when an objective non-informative prior is selected. For example, if we use non-informative prior $\pi(\mu) \propto 1$ for the mean parameter $\mu$ of the normal distribution in the following section \ref{simu1}, the values of $\log \pi(\hat \theta)$ and $E_{\theta | \textbf{y}}\log \pi (\theta )$ in equation (\ref{eq:Ando}) are undefined.

In order to avoid those drawbacks, we propose a new model selection criterion in terms of the posterior mean of the empirical log-likelihood $\hat{\eta}=\frac{1}{n}\sum_{i}E_{\theta | \textbf{y}}[\log g(y_{i}|\theta )]$, a natural estimator of $\eta $.  Without losing any of the attractive properties of BPIC, the new criterion expands the model scope to all Bayesian models under regularity conditions, improves the unbiased property for small samples, and enhances the robustness of the estimation.

Because all the data $\textbf{y}$ are used for both model fitting and model selection, $\hat\eta$ always overestimates $\eta$.  To correct the estimation bias from the overuse of the data, we propose the following theorem.

\begin{theorem}
\label{TM:PAIC}
Let $\textbf{y}=(y_1, y_2, \cdots, y_n)$ be $n$ independent observations drawn from the probability cumulative distribution $F(\tilde{y})$ with density function $f(\tilde{y})$.  Consider $\mathcal{G}=\{g(\tilde y| \theta );\theta \in \Theta \subseteq \mathbb{R}^{p}\}$ as a family of candidate statistical models that do not necessarily contain the true distribution $f$, where $\theta =(\theta _{1},...,\theta_{p})^{\prime}$ is the $p$-dimensional vector of unknown parameters, with prior distribution $\pi (\theta)$. 
Under the following three regularity conditions: 
\begin{description}
	\item{C1}: Both the log density function $\log g(\tilde y| \theta)$ and the log unnormalized posterior density $\log \{L(\theta | \textbf{y})\pi (\theta )\}$ are twice continuously differentiable in the compact parameter space $\Theta $;
	\item{C2}: The expected posterior mode $\theta_0 = \arg \max_ \theta E_{\tilde{y}}[\log \{g(\tilde{y}| \theta )\pi _{0}(\theta )\}]$ is unique in $\Theta $;
	\item{C3}: The Hessian matrix of $E_{\tilde{y}}[\log \{g(\tilde{y}| \theta )\pi _{0}(\theta )\}]$ is non-singular at $\theta_0$,\\ \vspace{-.1in}
\end{description}
 the bias of $\hat{\eta}$ for $\eta $ can be approximated asymptotically without bias by 
\begin{equation}
\hat{\eta}-\eta =\widehat{b_{\theta }}\cong \frac{1}{n}tr\{J_{n}^{-1}(\hat{\theta})I_{n}(\hat{\theta})\},
\end{equation}
where $\hat{\theta}$ is the posterior mode that minimizes the posterior distribution $\propto \pi (\theta) \prod_{i=1}^{n} g(y_i| \theta)$ and 
\begin{eqnarray*}
J_{n}(\theta ) &=&-\frac{1}{n}\sum_{i=1}^{n}(\frac{\partial ^{2}\log \{g(y_{i}| \theta )\pi^{\frac{1}{n}}(\theta )\}}{\partial \theta \partial \theta ^{\prime }}) \\
I_{n}(\theta ) &=&\frac{1}{n-1}\sum_{i=1}^{n}(\frac{\partial \log \{g(y_{i}| \theta )\pi^{\frac{1}{n}} (\theta )\}}{\partial \theta }\frac{\partial \log \{g(y_{i}| \theta )\pi ^{\frac{1}{n}}(\theta )\}}{\partial \theta ^{\prime }}).
\end{eqnarray*}
\end{theorem}

\begin{proof} Recall that the quantity of interest is $E_{\tilde{y}}E_{\theta |\textbf{y}}\log g(\tilde{y}|\theta )$.  To estimate it, we first check\\ $E_{\tilde{y}}E_{\theta |\textbf{y}}\log \{g(\tilde{y}|\theta )\pi
_{0}(\theta )\}=E_{\tilde{y}}E_{\theta |\textbf{y}}\{\log g(\tilde{y}|\theta )+\log
\pi _{0}(\theta )\}$ and expand it around $\theta _{0}$,%
\begin{eqnarray}
E_{\tilde{y}}E_{\theta |\textbf{y}}\log \{g(\tilde{y}|\theta )\pi _{0}(\theta )\}
&=&E_{\tilde{y}}\log \{g(\tilde{y}|\theta _{0})\pi _{0}(\theta
_{0})\}+E_{\theta |\textbf{y}}(\theta -\theta _{0})^{\prime }\frac{\partial E_{\tilde{%
y}}\log \{g(\tilde{y}|\theta )\pi _{0}(\theta )\}}{\partial \theta }%
|_{\theta =\theta _{0}}  \notag \\
&&+\frac{1}{2}E_{\theta |\textbf{y}}[(\theta -\theta _{0})^{\prime }\frac{\partial
^{2}E_{\tilde{y}}\log \{g(\tilde{y}|\theta )\pi _{0}(\theta )\}}{\partial
\theta \partial \theta ^{\prime }}|_{\theta =\theta _{0}}(\theta -\theta
_{0})]+o_{p}(n^{-1})  \notag \\
&=&E_{\tilde{y}}\log \{g(\tilde{y}|\theta _{0})\pi _{0}(\theta
_{0})\}+E_{\theta |\textbf{y}}(\theta -\theta _{0})^{\prime }\frac{\partial E_{\tilde{%
y}}\log \{g(\tilde{y}|\theta )\pi _{0}(\theta )\}}{\partial \theta }%
|_{\theta =\theta _{0}}  \notag \\
&&-\frac{1}{2}E_{\theta |\textbf{y}}[(\theta -\theta _{0})^{\prime }J(\theta
_{0})(\theta -\theta _{0})]+o_{p}(n^{-1})  \notag \\
&\triangleq &I_{1}+I_{2}+I_{3}+o_{p}(n^{-1})  \label{I123}
\end{eqnarray}

The first term $I_{1}$ can be linked to the empirical log likelihood
function as follows:%
\begin{eqnarray*}
E_{\tilde{y}}\log \{g(\tilde{y}|\theta _{0})\pi _{0}(\theta _{0})\} &=&E_{%
\tilde{y}}\log g(\tilde{y}|\theta _{0})+\log \pi _{0}(\theta _{0}) \\
&=&E_{y}\frac{1}{n}\log L(\theta _{0}|\textbf{y})+\log \pi _{0}(\theta _{0}) \\
&=&E_{y}\frac{1}{n}\log \{L(\theta _{0}|\textbf{y})\pi (\theta _{0})\}-\frac{1}{n}%
\log \pi (\theta _{0})+\log \pi _{0}(\theta _{0}) \\
&=&E_{y}E_{\theta |\textbf{y}}\frac{1}{n}\log \{L(\theta |\textbf{y})\pi (\theta )\}-\frac{1}{%
2n}tr\{J_{n}^{-1}(\theta _{0})I(\theta _{0})\} \\
&&+\frac{1}{2n}tr\{J_{n}^{-1}(\hat{\theta})J_{n}(\theta _{0})\}-\frac{1}{n}%
\log \pi (\theta _{0})+\log \pi _{0}(\theta _{0})+o_{p}(n^{-1})
\end{eqnarray*}%
where the last equation holds due to Lemma \ref{1e} (together with other Lemmas, provided in the Appendix).

The second term $I_{2}$ vanishes since 
\begin{equation*}
\frac{\partial E_{\tilde{y}}\log \{g(\tilde{y}|\theta )\pi _{0}(\theta )\}}{%
\partial \theta }|_{\theta =\theta _{0}}=0
\end{equation*}%
as $\theta _{0}$ is the expected posterior mode.

Using Lemma \ref{1d}, the third term $I_{3}$ can be rewritten as 
\begin{eqnarray*}
I_{3} &=&-\frac{1}{2}E_{\theta |\textbf{y}}(\theta -\theta _{0})^{\prime }J(\theta
_{0})(\theta -\theta _{0}) \\
&=&-\frac{1}{2}tr\{E_{\theta |\textbf{y}}[(\theta -\theta _{0})(\theta -\theta
_{0})^{\prime }]J(\theta _{0})\} \\
&=&-\frac{1}{2n}(tr\{J_{n}^{-1}(\theta _{0})I(\theta
_{0})J_{n}^{-1}(\theta _{0})J(\theta _{0})\}+tr\{J_{n}^{-1}(\hat{\theta}%
)J(\theta _{0})\})+o_{p}(n^{-1})
\end{eqnarray*}

By substituting each term in equation (\ref{I123}) and neglecting the
residual term, we obtain%
\begin{eqnarray*}
E_{\tilde{y}}E_{\theta |\textbf{y}}\log \{g(\tilde{y}|\theta )\pi _{0}(\theta )\}
&\simeq &E_{y}E_{\theta |\textbf{y}}\frac{1}{n}\log \{L(\theta |\textbf{y})\pi (\theta )\}-%
\frac{1}{2n}tr\{J_{n}^{-1}(\theta _{0})I(\theta _{0})\} \\
&&+\frac{1}{2n}tr\{J_{n}^{-1}(\hat{\theta})J_{n}(\theta _{0})\}-\frac{1}{n}%
\log \pi (\theta _{0})+\log \pi _{0}(\theta _{0}) \\
&&-\frac{1}{2n}(tr\{J_{n}^{-1}(\theta _{0})I(\theta
_{0})J_{n}^{-1}(\theta _{0})J(\theta _{0})\}+tr\{J_{n}^{-1}(\hat{\theta}%
)J(\theta _{0})\})
\end{eqnarray*}

Recall that we have defined $\log \pi _{0}(\theta)=\lim\nolimits_{n\rightarrow \infty }n^{-1}\log \pi (\theta )$, so that asymptotically we
have 
\begin{eqnarray*}
&\log \pi _{0}(\theta _{0})-\frac{1}{n}\log \pi (\theta _{0})\simeq 0,\\
&E_{\theta |\textbf{y}}\log \{\pi _{0}(\theta )\}-E_{\theta |\textbf{y}}\frac{1}{n}\log\{\pi (\theta )\}\simeq 0.
\end{eqnarray*}

Therefore, $E_{\tilde{y}}E_{\theta |\textbf{y}}\log \{g(\tilde{y}|\theta )\}$ can be
estimated by 
\begin{eqnarray*}
E_{\tilde{y}}E_{\theta |\textbf{y}}\log \{g(\tilde{y}|\theta )\} &=&E_{\tilde{y}%
}E_{\theta |\textbf{y}}\log \{g(\tilde{y}|\theta )\pi _{0}(\theta )\}-E_{\theta
|\textbf{y}}\log \{\pi _{0}(\theta )\} \\
&\simeq &E_{y}E_{\theta |\textbf{y}}\frac{1}{n}\log \{L(\theta |\textbf{y})\pi (\theta )\}-%
\frac{1}{2n}tr\{J_{n}^{-1}(\theta _{0})I(\theta _{0})\}+\frac{1}{2n}%
tr\{J_{n}^{-1}(\hat{\theta})J_{n}(\theta _{0})\} \\
&&-\frac{1}{2n}(tr\{J_{n}^{-1}(\theta _{0})I(\theta
_{0})J_{n}^{-1}(\theta _{0})J(\theta _{0})\}+tr\{J_{n}^{-1}(\hat{\theta}%
)J(\theta _{0})\}) \\
&&-\frac{1}{n}\log \pi (\theta _{0})+\log \pi _{0}(\theta _{0})-E_{\theta
|\textbf{y}}\log \{\pi _{0}(\theta )\} \\
&\simeq &E_{y}E_{\theta |\textbf{y}}\frac{1}{n}\log \{L(\theta |\textbf{y})\}-\frac{1}{2n}%
tr\{J_{n}^{-1}(\theta _{0})I(\theta _{0})\}+\frac{1}{2n}tr\{J_{n}^{-1}(%
\hat{\theta})J_{n}(\theta _{0})\} \\
&&-\frac{1}{2n}(tr\{J_{n}^{-1}(\theta _{0})I(\theta
_{0})J_{n}^{-1}(\theta _{0})J(\theta _{0})\}+tr\{J_{n}^{-1}(\hat{\theta}%
)J(\theta _{0})\})
\end{eqnarray*}

Replacing $\theta _{0}$ by $\hat{\theta}$, $J(\theta _{0})$ by $J_{n}(\hat{\theta})$ and $I(\theta _{0})$ by $I_{n}(\hat{\theta})$, we obtain $E_{\theta |\textbf{y}}\frac{1}{n}\log \{L(\theta |\textbf{y})\}-\frac{1}{n}tr\{J_{n}^{-1}(\hat{\theta})I_{n}(\hat{\theta})\}$ as an
asymptotically unbiased estimate for $E_{\tilde{y}}E_{\theta |\textbf{y}}\log \{g(%
\tilde{y}|\theta )\}$.
\end{proof}

With the above result, we propose a new predictive criterion for Bayesian modeling, the Posterior averaging information criterion (PAIC) as 
\begin{equation}
\textsc{PAIC} = -2\sum_i E_{\theta | \textbf{y}}[\log g(y_i| \theta )]+2 tr\{J_n^{-1}(\hat{\theta}) I_n(\hat{\theta})\}.
\end{equation}
The candidate models with small criterion values are preferred for the purpose of model selection.

The PAIC has several attractive properties. First, it assesses Bayesian model performance with respect to the posterior distribution function, which represents the current best knowledge from a Bayesian perspective in the family of candidate model. When the posterior distribution of the parameters is asymmetric, we expect it to better perform than any plug-in estimate based approaches. Secondly, it is an asymptotic unbiased estimator for the out-of-sample log-likelihood, a measure in terms of the Kullback-Leibler divergence for the similarity of the fitted model and the underlying true distribution.  %The estimation averaged over the posterior is more precise and robust than any plug-in–based estimators, especially when the posterior distribution of the parameters is asymmetric, which is a normal situation when the parameters are hierarchical. 
Thirdly, PAIC is derived free of the assumption that the approximating distributions contain the truth, indicating that PAIC is generally applicable even though some models could be mis-specified. Lastly, unlike the BPIC, the PAIC is well-defined and can cope with degenerate non-informative prior distributions for parameters. 

%To better understand the difference between PAIC and BPIC, we consider a simple example with the data generated from the normal distribution. The prior   

Several Bayesian researchers have also focused on the posterior averaged Kullback-Leibler discrepancy using cross-validation.  Plummer (2008) introduced the expected deviance penalized loss with `expected deviance' defined as $$L^e(y_i,z)= -2 E_{\theta| z} \log{g(y_i| \theta)},$$ which is a special case of the predictive discrepancy measure (Gelfand and Ghosh, 1998).  The standard cross-validation method can also be applied in this circumstance to estimate $\eta$, simply by considering the Kullback-Leibler discrepancy as the utility function of Vehtari and Lampinen (2002) and further investigated by Gelman et al. (2014).  The estimation of the bootstrap error correction $\eta ^{(b)}- \hat\eta ^{(b)}$ with bootstrap analogues $$\eta ^{(b)} =E_{\tilde{y^*}}[E_{\theta | y^*}\log g(\tilde{y}| \theta )]$$ and $$\hat \eta^{(b)} =E_{\tilde{y^*}}[n^{-1} E_{\theta | y^*}\log L(\theta | y^*)]$$ for $\eta - \hat\eta $ was discussed by Ando (2007) as a Bayesian adaptation of frequentist model selection (Konish and Kitagawa, 1996).  Although numeric approach such as importance sampling can be used for the intensive computation, one caveat is that it may cause inaccurate estimation in practice if some observation $y_i$ was influential (Vehtari and Lampinen, 2002).  To address that problem, Vehtari et al. (2017) proposed Pareto smoothed importance sampling, a new algorithm for regularizing importance weights, and developed a numerical tool (Vehtari et al., 2018) to facilitate computation.   Watanabe (2010) established a singular learning theory and proposed a new criterion named Watanabe-Akaike (Gelman et al., 2014), or widely applicable information criterion (WAIC; Watanabe 2008, 2009), while WAIC$_1$ was proposed for the plug-in discrepancy and WAIC$_2$ for the posterior averaged discrepancy.  However, compared with BPIC and PAIC, we found that WAIC$_2$ tends to have larger bias and variation for regular Bayesian models, as shown in simulation studies in the next section.

% The proposed theorem is implicitly conditional on the known values of %the covariates, X, if there is any.  The proof is an analogy.

%Extension :  Theorem 1 can be extended.  I agree with Box (1976) that %`all models are wrong, but some are useful'.  Whatever stategy is %employed, there will be a number of aspects in which the operating model %and the model ultimately fitted differ.  Each aspect of `lack of fit' %can be measured by some discrepancy, and the selection of the %discrepancy strongly depends on the specific purpose of the envisaged %statistical analysis.  Consequently, researchers should decide which %aspect the fitted model is required to best conform to the operating %model, and choose the corresponding discrepancy is to be minimized in %their own analysis.  Kullback-Leibler information quantity is a special %choice among the numerous discrepancy functions.  The reference on some %other important discrepancies in statistical literiture can be found in %Linhart and Zucchini (1986).

\section{Simulation Study}\label{sec4}

%\subsection{Preamble} 

In this section, we present numerical results to study the behavior of the proposed method under small and moderate sample sizes in both Gaussian and non-Gaussian settings. In the simulation experiments, we estimate the true expected bias $\eta $ either analytically ($\mathcal{x}$\ref{simu1}) or numerically by
averaging $E_{\theta | \textbf{y}}[\log g(\tilde{y}| \theta )]$ over a large number of extra independent draws of $\tilde y$ when there is no closed form for
the integration ($\mathcal{x}$\ref{simu2}). To have BPIC well-defined for comparison, only the proper prior distributions are considered.

\subsection{A case with closed-form expression for bias estimators} \label{simu1}

Suppose observations $\textbf{y}=(y_{1},y_{2},...,y_{n})$ are a vector of iid samples generated from $N(\mu _T,\sigma_T ^{2})$, with unknown true mean $\mu _{T}$ and variance $\sigma_T ^{2}=1$.  Assume the data are analyzed by the approximating model $g(y_{i}| \mu )=N(\mu ,\sigma_A ^{2})$ with prior $\pi (\mu )=N(\mu _{0},\tau _{0}^{2})$, where $\sigma_A^2$ is fixed, but not necessarily equal to the true variance $\sigma_T^2$. When $\sigma_A^2 \neq \sigma_T^2$, the model is misspecified.

\begin{figure}[]
%\captionsetup{width=1\textwidth}
\center
\includegraphics[width=6in]{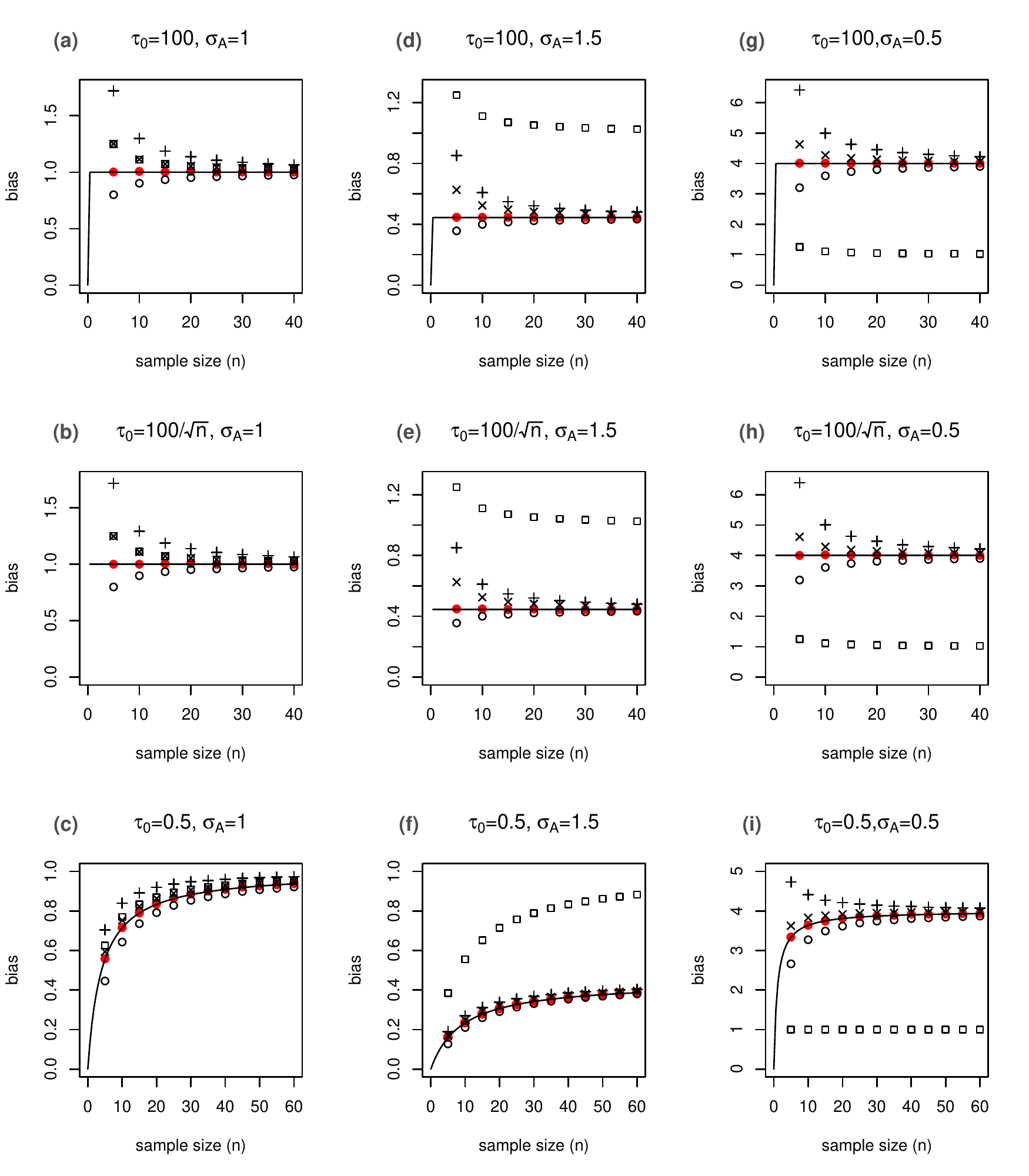}
\caption[Comparison of criteria with posterior averaging: $\sigma_A^2=\sigma_T ^2$]
{Performance of the bias estimators for $E_{y}(\hat{\eta}-\eta )$. The top panels are under a relatively non-informative prior with $\protect\tau_0^2=10^4$; the middle panels are under the case that the prior distribution grows with sample size with $\protect\tau_0^2=10^4/n$; the bottom panels are under an informative prior with $\protect\tau_0^2=0.25$. The left panels (a), (b) and (c) are under the scenario of $\sigma_A^2=\sigma_T ^2=1$, i.e., the true distribution is contained in the candidate models. The middle panels (d), (e) and (f) are under the scenario of $\sigma_A^2=2.25$ and right panels (g), (h) and (i) are under the scenario of $\sigma_A^2=0.25$ when the proposed model is misspecified from $\sigma_T ^2=1$. The true bias $b_\mu$ is curved by ( ----- ) as a function of sample size $n$. The averages of the different bias estimators are marked by ({\color{red}$\bullet$}) for PAIC; ($\circ$) for BPIC; (\footnotesize$\square$\normalsize) for $p_{opt}$; ($+$) for WAIC$_2$; and ($\times$) for cross-validation.  Each mark represents the mean of the estimated bias of 100,000 replications of $\textbf{y}$.  }
\label{fig:linearBias}
\end{figure}

The posterior distribution of $\mu$ is normally distributed with mean $\hat{\mu}$ and variance $\hat{\sigma}^{2}$, where 
\begin{eqnarray*}
\hat{\mu} &=&(\mu _{0}/\tau _{0}^{2}+\sum_{i=1}^{n}y_{i}/\sigma_A ^{2})/(1/\tau _{0}^{2}+n/\sigma_A ^{2}) \\
\hat{\sigma}^{2} &=&1/(1/\tau _{0}^{2}+n/\sigma_A ^{2}).
\end{eqnarray*}
Therefore, we obtain the Kullback-Leibler discrepancy function and its estimator as  
\begin{eqnarray*}
\eta &=&E_{\tilde{y}}[E_{\mu | \textbf{y}}[\log g(\tilde{y}| \mu )]]=-\frac{1}{2}\log (2\pi \sigma_A ^{2})-\frac{\sigma_T ^{2}+(\mu _{T}-\hat{\mu})^{2}+\hat{\sigma}^{2}}{2\sigma_A ^{2}} \\
\hat{\eta} &=&\frac{1}{n}\sum_{i=1}^{n}E_{\mu | \textbf{y}}[\log g(y_{i}| \mu )]]=-\frac{1}{2}\log (2\pi \sigma_A ^{2})-\frac{1}{n}\sum_{i=1}^{n}\frac{(y_{i}-\hat{\mu})^{2}+\hat{\sigma}^{2}}{2\sigma_A ^{2}}.
\end{eqnarray*}

To eliminate the estimation error caused by the sampling of the observations $\textbf{y}$, we average the bias $\hat{\eta}-\eta $ over $\textbf{y}$ with its true density $N(\mu_T,\sigma_T ^2)$, 
$$b_{\mu } =E_{\textbf{y}}(\hat{\eta}-\eta )=E_{\textbf{y}}\{\frac{\sigma_T^2}{2\sigma_A^2}+\frac{(\mu _{T}-\hat{\mu})^{2}}{2\sigma_A^2}-\frac{1}{n} \sum_{i=1}^{n}\frac{(y_{i}-\hat{\mu})^{2}}{2\sigma_A^2}\} = \sigma_T ^2 \hat{\sigma}^{2} / \sigma_A ^4.$$

Here, we compare the bias estimate defined in Theorem \ref{TM:PAIC}, $\hat{b}_\mu ^{PAIC} $ with 4 other bias estimators: 
$\hat{b}_\mu^{BPIC}$ (Ando, 2007), $\hat{b}_\mu^{WAIC_2}$ (Watanabe, 2009), $\hat{b}_\mu^{p_{opt}}$ (Plummer, 2008), and $\hat{b}_\mu^{CV}$ (Stone, 1974).
{%\small
\begin{eqnarray*}
\hat{b}_\mu^{PAIC} &=&\frac{1}{n-1}\hat{\sigma}^{2}\sum _{i=1}^{n}((\mu _{0}-\hat{\mu})/(n\tau _{0}^{2})+(y_{i}-\hat{\mu})/\sigma_A ^{2})^{2} \\ 
\hat{b}_\mu^{BPIC} &=&\frac{1}{n}\hat{\sigma}^{2}\sum _{i=1}^{n}((\mu_{0}-\hat{\mu})/(n\tau _{0}^{2})+(y_{i}-\hat{\mu})/\sigma_A ^{2})^{2}\\
%\hat{b}_\mu^{WAIC} &=&\frac{\hat\sigma^2}{\sigma_A^2} (1-\frac{\sum_{i=1}^n (y_{i}-\hat{\mu})^2/n}{\hat\sigma^2+\sigma_A^2}) - \log(1+\frac{\hat\sigma^2}{\sigma_A^2}) \\
\hat{b}_\mu^{WAIC_2} &=&\frac{\hat\sigma^2}{\sigma_A^4} (n\hat\sigma^2/2 + \sum_{i=1}^n (y_{i}-\hat{\mu})^2) \\
\hat{b}_\mu^{p_{opt}} &=&\frac 1 {2n} p_{opt} = 1/(1/\tau_0^2 + (n-1)/\sigma_A^2) / \sigma_A^2 \\
\hat{b}_\mu^{CV} &=& \hat \eta - (\sum_{i=1}^n (y_i-(\mu_0/\tau_0^2 + \sum_{j\neq i} y_j/\sigma_A^2)/(1/\tau_0^2 + (n-1)/\sigma_A^2))^2/n +\hat{\sigma}^2) / \sigma_A^2 /2. 
\end{eqnarray*}
}
The results are in accordance with the theory. All of the estimates are close to the true bias-correction values when the model is well-specified with $\sigma_A^2=\sigma_T^2=1$, especially when the sample size becomes moderately large (Figure \ref{fig:linearBias}, panels (a), (b) and (c)).  The estimated values based on the PAIC are consistently closer to the true values than either those based on Ando's method, which underestimates the bias, or the WAIC$_2$, cross-validation or expected deviance penalized loss, which overestimate the bias, especially when the sample size is small. When the models are misspecified, it is not surprising that in all of the plots given in panels (d)-(i) of Figure \ref{fig:linearBias}, only the expected deviance penalized loss misses the target even asymptotically since its assumption is violated, whereas all the other approaches converge to $b_{\mu }$.  In summary, PAIC achieves the best overall performance.

\subsection{Bayesian logistic regression}\label{simu2}
Consider frequencies $\textbf{y}=\{y_1,\ldots ,y_N\}$, which are independent observations from binomial distributions with respective true probabilities $\xi _1^T,\ldots ,\xi _N^T$, and sample sizes, $n_{1},\ldots ,n_{N}$.  To draw inference of the $\xi$'s, we assume that the logits 
\begin{equation*}
\beta_i=\logit(\xi_i)=\log \frac{\xi_i}{1-\xi _{i}}
\end{equation*}
are random effects that follow the normal distribution $\beta_{i}\sim N(\mu, \tau^2).$
The weakly-informative joint prior distribution $N(\mu; 0, 1000^2) \cdot Inv$-$\chi^2(\tau^2; 0.1, 10)$ is proposed on the hyper-parameter $(\mu ,\tau^2)$ so that the BPIC is properly defined and computable.  The posterior distribution is asymmetric due to the logistic transformation. 

%We also include $p_D$ in the comparison because when we follow the %arguments
We compare the performance of four asymptotically unbiased bias estimators in this hierarchical, asymmetric setting.  The true bias $\eta$ does not have an analytical form.  We estimate it through numerical computation using independent simulation from the same data generating process, assuming the underlying true values of $\mu=0$ and $\tau=1$. The simulation scheme is as follows:
\begin{enumerate}
	\item Draw $\beta_{T,i} \sim N(0,1)$, $y_i \sim Bin(n_i,\logit^{-1}(\beta_{T,i}))$, $i=1,\ldots,N$ from the true distribution. 
	\item Simulate the posterior draws of $(\beta, \mu, \tau)| \textbf{y}$.
	\item Estimate $\hat{b}_\beta^{PAIC}$, $\hat{b}_\beta^{BPIC}$, $\hat{b}_\beta^{WAIC_2}$ and $\hat{b}_\beta^{CV}$. 
	\item Draw $\textbf{z}^{(j)} \sim Bin(n,\logit^{-1}(\beta_0^T))$,  $j=1,\ldots,J$, for approximation of true $\eta$.
	\item Compare each $\hat{b}_\beta$ with true bias $b_\beta = \hat{\eta}-\eta$. 
 	\item Repeat steps 1-5.
\end{enumerate}
\vspace{-.1in}

\begin{table}[h]
 \caption[Criteria comparison for Bayesian logistic regression]{ The estimation error of bias correction: the mean and standard deviation (in parentheses) from 1000 replications.}
\label{table:LogitBias}
\begin{center}
\begin{small}
\begin{sc}
\vskip -0.2in
\begin{tabular}{ccccc}
\toprule
Criterion & Actual Error  & Mean Absolute Error & Mean Square Error \\
& $\hat{\eta}-\eta-\hat{b}_{\beta}$ & $\left| {\hat{\eta}-\eta-\hat{b}_{\beta}} \right| $ & $(\hat{\eta}-\eta-\hat{b}_{\beta} )^{2}$\\
\midrule
${PAIC}$   & \textbf{0.160 (0.238)} & \textbf{0.206 (0.199)} & \textbf{0.082 (0.207)} \\
${BPIC}$    &   0.259 (0.244)  &   0.272 (0.229)  &   0.127 (0.267) \\ 
%${p_{opt}}$    &   0.560 (0.240)  &   0.560 (0.240)  &   0.372 (0.375) \\ 
$CV$    &   0.840 (0.285)  &   0.840 (0.285)  &   0.786 (0.633) \\
$WAIC_2$    &   0.511 (0.248)  &   0.511 (0.248)  &   0.323 (0.389) \\
\bottomrule
%$\widehat{BPIC_{c}}$ &   0.43   &   0.52   &   0.5  \\ 
%         & ( 0.56 ) & ( 0.48 ) & ( 0.94 ) \\
%$p_{D}$    &   -0.3   &   0.67   &   0.85 \\
%         & ( 0.87 ) & ( 0.63 ) & ( 1.96 )
\end{tabular}
\end{sc}
\end{small}
\end{center}
\vskip -0.1in
\end{table}

Table \ref{table:LogitBias} summarizes the bias and standard deviation of the estimation error when we choose $N=15$ and $n_{1}=\ldots =n_{N}=50$, and the $\beta$'s are independently simulated from the standard normal distribution assuming the true hyper-parameter mean $\mu=0$ and variance $\tau^2=1$.  The simulation is repeated for $1,000$ scenarios, each with $J=20,000$ for out-of-sample $\eta$ estimation.  PAIC and BPIC were calculated based on definition; leave-one-out cross-validation and WAIC$_2$ were estimated using \textit{R} package \textit{loo} v2.0.0.  The actual error, mean absolute error and mean square error were considered to assess the estimation error using the bias correction estimates.  With respect to all three different metrics, the bias estimation of PAIC is consistently superior to other methods.  %BPIC performs better than $p_{D}$, The improvement compared to BPIC, %which performs better than $p_{D}$, is very large. 
Compared to BPIC, the second best performed model selection criterion, the bias and the mean squared error of PAIC are reduced by about $40\%$, while the absolute bias is reduced by about one quarter, which matches our expectation that the natural estimate $\frac 1 n \sum_i E_{\theta | \textbf{y}}[\log g(y_i| \theta )]$ will estimate the posterior averaged Kullback-Leibler discrepancy more precisely than plug-in estimate $\frac 1 n \sum_i \log g(y_i| \hat\theta )$ when the posterior distribution is asymmetric and correlated.  Compared to WAIC$_2$, the bias, absolute error and mean square error of PAIC are dramatically reduced by at least $60\%$.  In practice, we expect the improvement is even larger when proposed models were more complicated in terms of hierarchical structures.

\section{Discussions and concluding remarks}
\label{sec6}

The Kullback-Leilber divergence is a non-symmetric measure of the difference between two probability distributions.  Frequentist statistics theoretically employing Kullback-Leibler divergence into parametric model selection emerged during the 1970s.  Since then, the development of related theory and applications has rapidly accelerated.

Bayesian model selection in terms of the Kullback-Leibler divergence has drawn substantial attention in the past two decades.  The availability of both fast computers and advanced numerical methods enables the empirical popularity of Bayesian modeling, which allows for additional flexibility to incorporate the information out of the data, as represented by the prior distribution.  The fundamental assumption of Bayesian inference is different from frequentist statistics, for the unknown parameters are treated as random variables in the form of a probability distribution.  Taking this into account, it is important to have selection techniques that are specifically designed for Bayesian modeling.  

Before the proposal of any specific model selection criterion, two questions should be first investigated to guide the method development. 1. What is a natural Kullback-Leibler discrepancy to evaluate a Bayesian model? 2. What is a good estimate for Kullback-Leibler discrepancy for Bayesian model?  The prevailing plug-in parameter methods, such as DIC, presume the candidate models are correct, and assess the goodness of each candidate model with a density function specified by the plug-in parameters.  However, from a Bayesian perspective, it is inherent to examine the performance of a Bayesian model over the entire posterior distribution, as stated by Celeux et al. (2006, p.703): ``... we concede that using a plug-in estimate disqualifies the technique from being properly Bayesian.''  Accordingly, statistical approaches to estimate the Kullback-Leibler discrepancy as evaluated by averaging over the posterior distribution are of great interest.

We have proposed PAIC, a versatile model selection technique for Bayesian models under regularity assumptions, to address this problem.  From a predictive perspective, we consider the asymptotic unbiased estimation of a Kullback-Leibler discrepancy, which averages the conditional density of the observable data against the posterior knowledge about the unobservable data.  Empirically, the proposed PAIC measures the similarity of the fitted model and the underlying true distribution, regardless of whether or not the approximating distribution family contains the true model.  %The implementation of PAIC is relatively simple and the computational cost is low.  
%In the areas of genetics, sports, ecology, psychology, sociology and %political science, t
The range of applications of the proposed criterion can be quite broad.

PAIC and BPIC (Ando, 2007) are similar in many aspects.  In addition to all the good properties both methods share, PAIC has some special features, mainly because of the natural log-likelihood estimator.  For example, PAIC can be well applied even if the prior distribution of the parameters degenerates, in which case BPIC becomes uninterpretable.  A non-informative prior appears quite often in practice.  When the posterior distribution is asymmetric or parameters were correlated, our method provides a better bias estimation than that obtained by using the natural estimator, which evaluates the log-likelihood over the posterior distribution instead of over some specific point.  In numerical experiments to compare the performance of the proposed criteria with other Bayesian model selection criteria including BPIC and WAIC$_2$, PAIC has the smallest bias and variance to estimate the posterior averaged discrepancy. Even for data obtained from small sample sizes, the bias correction of PAIC still achieves better performance.

There are some future directions of the current work.  A more comprehensive comparison of Bayesian predictive methods for model selection can be investigated by taking into account the likely over-fitting in the selection phase, similar to Piironen and Vehtari (2017).  Because it's inconvenient that the users of PAIC and BPIC have to specify the first and second derivatives of the posterior distribution in their modeling, development of advanced computational tools for simultaneous calculation are really in need.  In singular learning machines, the regularity conditions can be relaxed to singular in a sense that the mapping from parameters to probability distributions is not necessarily one-to-one.  Although here we focused on only the regular models, it is also possible to generalize PAIC to singular settings with a modified bias correction term, after an algebraic geometrical transformation of the singular parameter space to a real $d$-dimensional manifold.

\bigskip
\begin{center}
{\large\bf SUPPLEMENTAL MATERIALS}
\end{center}

\section*{Lemmas for Proof of Theorem \ref{TM:PAIC}}

\subsection*{Some important notations}

 By the law of large numbers we have $\frac{1}{n}%
\log \{L(\theta |\textbf{y})\pi (\theta )\}\rightarrow E_{\tilde{y}}[\log \{g(\tilde{y%
}|\theta )\pi _{0}(\theta )\}]$ as $n$ tends to infinity. Denote $\theta
_{0} $, $\hat{\theta}$ the expected and empirical posterior mode of the log
unnormalized posterior density $\log \{L(\theta |\textbf{y})\pi (\theta )\}$, $i.e.$, 
\begin{eqnarray*}
\theta _{0} &=&\arg \max_{\theta }E_{\tilde{y}}[\log \{g(\tilde{y}|\theta
)\pi _{0}(\theta )\}] \\
\hat{\theta} &=&\arg \max_{\theta }\frac{1}{n}\log \{L(\theta |\textbf{y})\pi (\theta
)\},
\end{eqnarray*}
and let $I(\theta )$ and $J(\theta )$ denote the Bayesian Hessian matrix and Bayesian Fisher information matrix
\begin{equation*}
I(\theta )=E_{\tilde{y}}\left(\frac{\partial \log \{g(\tilde{y}|\theta )\pi
_{0}(\theta )\}}{\partial \theta }\frac{\partial \log \{g(\tilde{y}|\theta
)\pi _{0}(\theta )\}}{\partial \theta ^{\prime }}\right)
\end{equation*}%
and 
\begin{equation*}
J(\theta )=-E_{\tilde{y}}\left(\frac{\partial ^{2}\log \{g(\tilde{y}|\theta )\pi
_{0}(\theta )\}}{\partial \theta \partial \theta ^{\prime }}\right).
\end{equation*}

\subsection*{Proof of Lemmas}

We start with a few lemmas to support the proofs of Theorem \ref{TM:PAIC}.

\begin{lemma}\label{1a}
Under the same regularity conditions of Theorem \ref{TM:PAIC}, $\sqrt{n}(\hat{\theta}-\theta
_{0})$ is asymptotically distributed as $N(0,J_{n}^{-1}(\theta_{0})I(\theta _{0})J_{n}^{-1}(\theta _{0}))$.
\end{lemma}

\begin{proof}
Consider the Taylor expansion of $\frac{\partial \log
\{L(\theta |\textbf{y})\pi (\theta )\}}{\partial \theta }|_{\theta =\hat{\theta}}$ at 
$\theta _{0}$, 
\begin{eqnarray}
\frac{\partial \log \{L(\theta |\textbf{y})\pi (\theta )\}}{\partial \theta }|_{\theta =\hat{\theta}} &\simeq &\frac{\partial \log \{L(\theta |\textbf{y})\pi
(\theta )\}}{\partial \theta }|_{\theta =\theta _{0}}+\frac{\partial^{2}\log \{L(\theta |\textbf{y})\pi (\theta )\}}{\partial \theta \partial \theta^{\prime }}|_{\theta =\theta _{0}}(\hat{\theta}-\theta _{0}) \nonumber \\
&=&\frac{\partial \log \{L(\theta |\textbf{y})\pi (\theta )\}}{\partial \theta }|_{\theta =\theta _{0}}-nJ_{n}(\theta _{0})(\hat{\theta}-\theta _{0}). \nonumber
\end{eqnarray}%
Note that $\hat{\theta}$ is the mode of $\log \{L(\textbf{y}|\theta )\pi (\theta )\}$ and satisfies $\frac{\partial \log \{L(\textbf{y}|\theta )\pi (\theta )\}}{\partial \theta }|_{\theta =\hat{\theta}}=0$. Plug it into the above equation, we have 
\begin{equation}
nJ_{n}(\theta _{0})(\hat{\theta} - \theta _{0})\simeq \frac{\partial \log
\{L(\theta |\textbf{y})\pi (\theta )\}}{\partial \theta }|_{\theta =\theta _{0}}. \label{A1}
\end{equation}
From the central limit theorem, the right-hand-side (RHS) of the equation (\ref{A1}) is approximately
distributed as $N(0,nI(\theta _{0}))$ when $E_{y}\frac{\partial \log
\{L(\theta |\textbf{y})\pi (\theta )\}}{\partial \theta }|_{\theta =\theta _{0}}\rightarrow 0$.
Therefore 
\begin{equation*}
\sqrt{n}(\hat{\theta}-\theta _{0})\sim N(0,J_{n}^{-1}(\theta
_{0})I(\theta _{0})J_{n}^{-1}(\theta _{0})).
\end{equation*} %\QEDA
\end{proof}

\begin{lemma}\label{1b}
Under the same regularity conditions of Theorem \ref{TM:PAIC}, $\sqrt{n}(\theta -\hat{\theta})\sim N(0,J_{n}^{-1}(\hat{%
\theta}))$.
\end{lemma}

\begin{proof}
Taylor-expand the logarithm of $L(\theta |\textbf{y})\pi (\theta )$\
around the posterior mode $\hat{\theta}$%
\begin{equation}
\log L(\theta |\textbf{y})\pi (\theta )=\log L(\hat{\theta}|\textbf{y})\pi (\hat{\theta})-%
\frac{1}{2}(\theta -\hat{\theta})^{\prime }\frac{1}{n}J_{n}^{-1}(\hat{\theta})(\theta -\hat{\theta})+o_p(n^{-1}) \label{A2}
\end{equation}%
where $J_{n}(\hat{\theta})=-\frac{1}{n}\frac{\partial ^{2}\log \{L(\theta
|\textbf{y})\pi (\theta )\}}{\partial \theta \partial \theta ^{\prime }}|_{\theta =%
\hat{\theta}}$.

Consider the RHS of equation (\ref{A2}) as a function of $\theta $: the first term  is a constant,
whereas the second term is proportional to the logarithm of a normal density.  It yields the approximation of the posterior distribution for $\theta$: 
\begin{equation*}
p(\theta |\textbf{y})\approx N(\hat{\theta},\frac{1}{n}J_{n}^{-1}(\hat{\theta})),
\end{equation*}
which completes the proof. 

Alternatively, though less intuitive, this lemma can also be proved by applying the Berstein-Von Mises theorem.%\QEDA
\end{proof}

\begin{lemma}\label{1c}
Under the same regularity conditions of Theorem \ref{TM:PAIC}, $E_{\theta |\textbf{y}}(\theta _{0}-\hat{\theta})(\hat{\theta}%
-\theta )^{\prime }=o_{p}(n^{-1})$.
\end{lemma}

\begin{proof}
First we have 
\begin{equation*}
\frac{\partial \log \{L(\theta |\textbf{y})\pi (\theta )\}}{\partial \theta }=\frac{%
\partial \log \{L(\theta |\textbf{y})\pi (\theta )\}}{\partial \theta }|_{\theta =%
\hat{\theta}}-nJ_{n}(\hat{\theta})(\theta -\hat{\theta})+O_{p}(1).
\end{equation*}

Since $\hat{\theta}$ is the mode of $\log \{L(\theta |\textbf{y})\pi (\theta )\}$, it
satifies $\frac{\partial \log \{L(\theta |\textbf{y})\pi (\theta )\}}{\partial \theta 
}|_{\theta =\hat{\theta}}=0$. Therefore $(\hat{\theta}-\theta
)=n^{-1}J_{n}^{-1}(\hat{\theta})\frac{\partial \log \{L(\theta |\textbf{y})\pi
(\theta )\}}{\partial \theta }+O_{p}(n^{-1})$. Note that 
\begin{eqnarray}
E_{\theta |\textbf{y}}\frac{\partial \log \{L(\theta |\textbf{y})\pi (\theta )\}}{\partial
\theta } &=&\int \frac{\partial \log \{L(\theta |\textbf{y})\pi (\theta )\}}{\partial
\theta }\frac{L(\theta |\textbf{y})\pi (\theta )}{p(\textbf{y})}d\theta  \nonumber \\
&=&\int \frac{1}{L(\theta |\textbf{y})\pi (\theta )}\frac{\partial \{L(\theta |\textbf{y})\pi
(\theta )\}}{\partial \theta }\frac{L(\theta |\textbf{y})\pi (\theta )}{p(\textbf{y})}d\theta
\nonumber \\
&=&\frac{1}{p(\textbf{y})} \int \frac{\partial \{L(\theta |\textbf{y})\pi(\theta )\}}{\partial \theta }d\theta
\nonumber \\
&=&\frac{1}{p(\textbf{y})} \frac{\partial }{\partial \theta }\int L(\theta |\textbf{y})\pi (\theta )d\theta =\frac{\partial }{\partial \theta }1=0. \nonumber 
\end{eqnarray}%
Because of assumption (C1), the equation holds when we change the order of the
integral and derivative. Therefore 
\begin{equation*}
E_{\theta |\textbf{y}}(\hat{\theta}-\theta )=n^{-1}J_{n}^{-1}(\hat{\theta})E_{\theta
|\textbf{y}}\frac{\partial \log \{L(\theta |\textbf{y})\pi (\theta )\}}{\partial \theta }%
+O_{p}(n^{-1})=O_{p}(n^{-1}).
\end{equation*}%
Together with $\theta _{0}-\hat{\theta}=$ $O_{p}(n^{-1/2})$ derived from
Lemma \ref{1a}, we complete the proof. %\QEDA
\end{proof}

\begin{lemma}\label{1d} 
Under the same regularity conditions of Theorem \ref{TM:PAIC}, $E_{\theta |\textbf{y}}(\theta _{0}-\theta )(\theta _{0}-\theta
)^{\prime }=\frac{1}{n}J_{n}^{-1}(\hat{\theta})+\frac{1}{n}J_{n}^{-1}(\theta
_{0})I(\theta _{0})J_{n}^{-1}(\theta _{0})+o_{p}(n^{-1}).$
\end{lemma}

\begin{proof}
$E_{\theta |\textbf{y}}(\theta _{0}-\theta )(\theta _{0}-\theta
)^{\prime }$ can be rewritten as $(\theta _{0}-\hat{\theta})(\theta _{0}-%
\hat{\theta})^{\prime }+E_{\theta |\textbf{y}}(\hat{\theta}-\theta )(\hat{\theta}%
-\theta )^{\prime }+2E_{\theta |\textbf{y}}(\theta _{0}-\hat{\theta})(\hat{\theta}%
-\theta )$. Applying Lemmas \ref{1a}, \ref{1b} and \ref{1c}, we complete the proof.  %\QEDA
\end{proof}

\begin{lemma}\label{1e} 
Under the same regularity conditions of Theorem \ref{TM:PAIC}, 
\begin{eqnarray*}
E_{\theta |\textbf{y}}\frac{1}{n}\log \{L(\textbf{y}|\theta )\pi (\theta )\} &\simeq &\frac{1}{%
n}\log \{L(\theta _{0}|\textbf{y})\pi (\theta _{0})\} \\
&&+\frac{1}{2n}(tr\{J_n^{-1}(\theta _{0})I(\theta _{0})\}-tr\{J_{n}^{-1}(\hat{%
\theta})J_{n}(\theta _{0})\})+O_{p}(n^{-1}).
\end{eqnarray*}
\end{lemma}

\begin{proof}
The posterior mean of the log joint density distribution of $(\textbf{y},\theta )$ can be Taylor-expanded around $\theta _{0}$ as%
\begin{eqnarray}
E_{\theta |\textbf{y}}\frac{1}{n}\log \{L(\theta |\textbf{y})\pi (\theta )\} &=&\frac{1}{n}%
\log \{L(\theta _{0}|\textbf{y})\pi (\theta _{0})\}+E_{\theta |\textbf{y}}(\theta -\theta
_{0})^{\prime }\frac{1}{n}\frac{\partial \log \{L(\theta |\textbf{y})\pi (\theta )\}}{%
\partial \theta }|_{\theta =\theta _{0}}  \notag \\
&&+\frac{1}{2}E_{\theta |\textbf{y}}(\theta -\theta _{0})^{\prime }\frac{1}{n}\frac{%
\partial ^{2}\log \{L(\theta |\textbf{y})\pi (\theta )\}}{\partial \theta \partial
\theta ^{\prime }}|_{\theta =\theta _{0}}(\theta -\theta
_{0})+o_{p}(n^{-1})  \notag \\
&=&\frac{1}{n}\log \{L(\theta _{0}|\textbf{y})\pi (\theta _{0})\}+E_{\theta
|\textbf{y}}(\theta -\theta _{0})^{\prime }\frac{1}{n}\frac{\partial \log \{L(\theta
|\textbf{y})\pi (\theta )\}}{\partial \theta }|_{\theta =\theta _{0}}  \notag \\
&&-\frac{1}{2}E_{\theta |\textbf{y}}(\theta -\theta _{0})^{\prime }J_{n}(\theta
_{0})(\theta -\theta _{0})+o_{p}(n^{-1}).  \label{1e1}
\end{eqnarray}%

Expand $\frac{\partial \log \{L(\theta |\textbf{y})\pi (\theta )\}}{\partial \theta }|_{\theta=\hat{\theta}}$ around $\theta _{0}$ to the first order, we obtain 
\begin{equation}
\frac{\partial \log \{L(\theta |\textbf{y})\pi (\theta )\}}{\partial \theta }|_{\theta =\hat{\theta}}=\frac{\partial \log \{L(\theta |\textbf{y})\pi (\theta )\}}{\partial \theta }|_{\theta =\theta _{0}}-nJ_{n}(\theta _{0})(\hat{\theta}-\theta _{0})+O_{p}(n^{-1}). \label{A4}
\end{equation} 
Because the posterior mode $\hat{\theta}$ is the solution of $\frac{\partial \log \{L(\theta |\textbf{y})\pi (\theta )\}}{\partial \theta }=0$, the equation (\ref{A4}) can be re-written as  
\begin{equation*}
\frac{1}{n}\frac{\partial \log \{L(\theta |\textbf{y})\pi (\theta )\}}{\partial
\theta }|_{\theta =\theta _{0}}=J_{n}(\theta _{0})(\hat{\theta}-\theta
_{0})+O_{p}(n^{-1}).
\end{equation*}%
Substituting it into the second term of (\ref{1e1}), the expansion of $%
E_{\theta |\textbf{y}}\frac{1}{n}\log \{L(\theta |\textbf{y})\pi (\theta )\}$ becomes:%
\begin{eqnarray*}
E_{\theta |\textbf{y}}\frac{1}{n}\log \{L(\theta |\textbf{y})\pi (\theta )\} &=&\frac{1}{n}%
\log \{L(\theta _{0}|\textbf{y})\pi (\theta _{0})\}+E_{\theta |\textbf{y}}(\theta -\theta
_{0})^{\prime }J_{n}(\theta _{0})(\hat{\theta}-\theta _{0}) \\
&&-\frac{1}{2}E_{y}E_{\theta |\textbf{y}}(\theta -\theta _{0})^{\prime }J_{n}(\theta
_{0})(\theta -\theta _{0})+o_{p}(n^{-1}) \\
&=&\frac{1}{n}\log \{L(\theta _{0}|\textbf{y})\pi (\theta _{0})\}+tr\{E_{\theta |\textbf{y}}[(%
\hat{\theta}-\theta _{0})(\theta -\theta _{0})^{\prime }]J_{n}(\theta _{0})\}
\\
&&-\frac{1}{2}tr\{E_{\theta |\textbf{y}}[(\theta -\theta _{0})(\theta -\theta
_{0})^{\prime }]J_{n}(\theta _{0})\}+o_{p}(n^{-1}) \\
&=&\frac{1}{n}\log \{L(\theta _{0}|\textbf{y})\pi (\theta _{0})\}+tr\{E_{\theta
|\textbf{y}}[(\theta -\theta _{0})(\hat{\theta}-\theta _{0})^{\prime }]J_{n}(\theta
_{0})\} \\
&&-\frac{1}{2}tr\{\frac{1}{n}[J_{n}^{-1}(\hat{\theta})+J_{n}^{-1}(\theta
_{0})I(\theta _{0})J_{n}^{-1}(\theta _{0})]J_{n}(\theta
_{0})\}+o_{p}(n^{-1})
\end{eqnarray*}%
where in the last line we replace $E_{\theta |\textbf{y}}[(\theta -\theta
_{0})(\theta -\theta _{0})^{\prime }]$ with the result of Lemma \ref{1d}. $%
E_{\theta |\textbf{y}}[(\theta -\theta _{0})(\hat{\theta}-\theta _{0})^{\prime }]$ in
the second term of the expansion can be rewritten as $E_{\theta |\textbf{y}}[(\hat{%
\theta}-\theta _{0})(\hat{\theta}-\theta _{0})^{\prime }]+E_{\theta
|\textbf{y}}[(\theta -\hat{\theta})(\hat{\theta}-\theta _{0})^{\prime }]$, where the
former term is asymptotically equal to $\frac{1}{n}J_{n}^{-1}(\theta
_{0})I(\theta _{0})J_{n}^{-1}(\theta _{0})$ by Lemma \ref{1a}, and the latter
is negligible with higher order $o_{p}(n^{-1})$, as shown in Lemma \ref{1c}.
Therefore, the expansion can be finally simplified as%
\begin{eqnarray}
E_{\theta |\textbf{y}}\frac{1}{n}\log \{L(\textbf{y}|\theta )\pi (\theta )\} &\simeq &\frac{1}{%
n}\log \{L(\theta _{0}|\textbf{y})\pi (\theta _{0})\} \nonumber\\
&&+\frac{1}{2n}(tr\{J_n^{-1}(\theta _{0})I(\theta _{0})\}-tr\{J_{n}^{-1}(\hat{%
\theta})J_{n}(\theta _{0})\})+O_{p}(n^{-1}). \nonumber
\end{eqnarray}%\QEDA
\end{proof}

\end{document}